\documentclass[12pt,nosubfloats]{clear2022-arxiv} 
\usepackage{amsmath,amsfonts,amssymb}
\usepackage{bbm}
\usepackage{bm}
\usepackage{xcolor}

\newtheorem{defn}{Definition}
\newtheorem{thm}{Theorem}
\newtheorem{prop}{Proposition}
\newtheorem{cor}{Corollary}

\newtheorem*{question*}{Question}
\newtheorem*{answer*}{Answer}
\newtheorem*{solution*}{Solution}
\newtheorem*{nextstep*}{Next Step}
\newtheorem*{issue*}{Issue}

\usepackage{prettyref}
\newcommand{\rref}[2][]{\prettyref{#2}}
\newrefformat{model}{Model\,\ref{#1}}
\newrefformat{listing}{Listing\,\ref{#1}}
\newrefformat{alg}{Algorithm\,\ref{#1}}
\newrefformat{line}{line\,\ref{#1}}
\newrefformat{sec}{Section\,\ref{#1}}
\newrefformat{subsec}{Section\,\ref{#1}}
\newrefformat{section}{Section\,\ref{#1}}
\newrefformat{appendix}{Appendix\,\ref{#1}}
\newrefformat{app}{Appendix\,\ref{#1}}
\newrefformat{def}{Definition\,\ref{#1}}
\newrefformat{defn}{Definition\,\ref{#1}}
\newrefformat{thm}{Theorem\,\ref{#1}}
\newrefformat{ax}{\ref{#1}}
\newrefformat{prop}{Prop.\,\ref{#1}}
\newrefformat{lemma}{Lemma\,\ref{#1}}
\newrefformat{cor}{Corollary\,\ref{#1}}
\newrefformat{corollary}{Corollary\,\ref{#1}}
\newrefformat{ex}{Example\,\ref{#1}}
\newrefformat{tab}{Table\,\ref{#1}}
\newrefformat{fig}{Fig.\,\ref{#1}}
\newrefformat{eqn}{Equation~(\ref{#1})}
\newrefformat{problem}{Problem\,\ref{#1}}
\newrefformat{assumption}{Assumption\,\ref{#1}}
\newrefformat{claim}{Claim\,\ref{#1}}
\newrefformat{remark}{Remark\,\ref{#1}}

\newcommand{\Var}{\textrm{Var}}

\newcommand{\Unif}{\mathsf{Unif}}

\DeclareMathOperator\arctanh{arctanh}

\DeclareMathOperator{\pa}{pa}
\DeclareMathOperator{\ch}{ch}

\newcommand{\bbR}{\mathbb{R}}

\newcommand{\bbP}{\mathbb{P}}

\newcommand{\bbX}{\mathbb{X}}

\newcommand{\hSigma}{\hat{\Sigma}}

\newcommand{\hsigma}{\hat{\sigma}}

\newcommand{\cG}{\mathcal{G}}

\newcommand{\cN}{\mathcal{N}}
\newcommand{\cO}{\mathcal{O}}

\newcommand{\sidak}{\textrm{sidak}}

\usepackage{algorithm}
\usepackage{algorithmic}
\usepackage{booktabs}
\usepackage{wrapfig}
\usepackage{subcaption}
\usepackage{thmtools,thm-restate}

\newcommand{\rank}{{\textrm{rank}}}

\title[Learning Latent Factor Causal Models]{Causal Structure Discovery between Clusters \\of Nodes Induced by Latent Factors}
\usepackage{times}

\clearauthor{%
 \Name{Chandler Squires*} \Email{csquires@mit.edu}\\
 \addr LIDS, IDSS, and CSAIL, MIT, Cambridge, MA, USA
 \AND
 \Name{Annie Yun*} \Email{annieyun@mit.edu}\\
 \addr LIDS and IDSS, MIT, Cambridge, MA, USA
 \AND
 \Name{Eshaan Nichani} \Email{eshnich@mit.edu}\\
 \addr LIDS and IDSS, MIT, Cambridge, MA, USA
 \AND
 \Name{Raj Agrawal} \Email{r.agrawal@csail.mit.edu}\\
 \addr LIDS, IDSS, and CSAIL, MIT, Cambridge, MA, USA
 \AND
 \Name{Caroline Uhler} \Email{cuhler@mit.edu}\\
 \addr LIDS and IDSS, MIT, and Broad Institute,  Cambridge, MA, USA
}

\begin{document}

\maketitle

\begin{abstract}
  We consider the problem of learning the structure of a causal directed acyclic graph (DAG) model in the presence of latent variables.
  We define \textit{latent factor causal models} (LFCMs) as a  restriction on causal DAG models with latent variables, which are composed of clusters of observed variables that share the same latent parent and connections between these clusters given by edges pointing from the observed variables to latent variables.
  LFCMs are motivated by gene regulatory networks, where regulatory edges, corresponding to transcription factors, connect spatially clustered genes.
  We show identifiability results on this model and design a consistent three-stage algorithm that discovers clusters of observed nodes, a partial ordering over clusters, and finally, the entire structure over both observed and latent nodes. 
  We evaluate our method in a synthetic setting, demonstrating its ability to almost perfectly recover the ground truth clustering even at relatively low sample sizes, as well as the ability to recover a significant number of the edges from observed variables to latent factors. 
  Finally, we apply our method in a semi-synthetic setting to protein mass spectrometry data with a known ground truth network, and achieve almost perfect recovery of the ground truth variable clusters.
\end{abstract}


\section{Introduction}\label{sec:intro}

Structural causal models are valuable tools for reasoning about decision-making, and as a result, have been widely adopted across fields such as genomics \citep{friedman2000using}, econometrics \citep{blalock2017causal}, and epidemiology \citep{robins2000marginal}.
To use causal models when the causal structure is not known \textit{a priori}, it is necessary to learn the model from observed data, a task known as \textit{causal structure learning} \citep{heinze2018causal}.
As a field, causal structure learning has recently experienced major developments and remains an active and widespread area of research.
Recent works aim to address a number of challenges inherent to the problem of learning causal structure, such as the presence of unobserved confounders \citep{cai2019triad,frot2019robust,bernstein2020ordering}, the large search space over causal models \citep{chickering2002optimal,solus2021consistency}, identifiability of the underlying causal model \citep{shimizu2006linear,peters2014identifiability}, and statistical issues stemming from high-dimensional datasets \citep{nandy2018high}.
We focus on a setting which exhibits all of these challenges, and our proposed method addresses each of these challenges in a cohesive way.
We devote particular attention to the issue of unobserved confounders, 

A number of methods have been proposed to address the challenge of learning causal models in the presence of unobserved confounders.
These methods fall into two general categories.
First, some methods account for unobserved confounders by learning a graphical model over only the observed variables, albeit from a \textit{different} class of graphical models \citep{richardson2002ancestral,bernstein2020ordering}. 
However, in some cases, such as the one explored in this paper, it is possible to learn a graph over \textit{both} the observed and latent variables. 
Existing methods \citep{silva2006learning,kummerfeld2016causal,xie2020generalized,agrawal2021decamfounder} that seek to recover these structures often assume that the latent variables are \textit{exogenous}, i.e., are not caused by any of the observed variables.
However, this assumption is often violated in many applications.
For example, in genomics, gene regulatory networks are often modeled using \textit{transcription factories} \citep{stadhouders2019transcription} as underlying latent variables with gene expression as the observable variables.
These gene expressions then can have downstream impacts on other transcription factories, requiring a model that allows non-exogenous latent variables. 

\textbf{Contributions.} In \rref{sec:setup}, we introduce the class of \textit{latent factor causal models} (LFCMs), which allow for non-exogenous latent variables.
Similar to prior work, this class of models prohibits direct edges between observed variables, i.e., the effect of one observed variable on another must be mediated by some latent variable.
We likewise prohibit direct edges between latent variables, inducing a \textit{bipartite} structure over the graph.
Furthermore, we require that the latent variables \textit{cluster} the observed variables, i.e., each observed variable has only a single latent parent, and that each latent variable has at least three observed children.
These constraints on the model are motivated by how the DNA is organized in the cell nucleus to facilitate cell-type specific gene expression. The spatial clustering of genes in the cell nucleus facilitates their co-regulation by transcription factors~\citep{Belyaeva, Shiva_review}. The expression of each gene represents the observed variables, the spatial clustering of genes is unobserved, 
and the latent factors represent the presence of transcription factors that can e.g.~turn on the expression of the co-clustered genes.
%

In \rref{sec:identifiability}, we establish identifiability results for LFCMs, based primarily on the tetrad representation theorem of \cite{spirtes2013calculation}. 
Based on our identifiability results, in \rref{sec:methods} we propose a constraint-based method for learning the underlying graph over \textit{both} latent and observed variables.
The proposed method has three stages.
In the first stage, our method identifies clusters of observed variables with the same latent parent, as well as an ordering over these clusters.
The second stage merges clusters from the first stage if necessary.
In the third stage, we learn edges from the observed variables to the latent variables, by testing for conditional independence with all children of each latent variable.
For each stage, the constraints being checked are equivalent to multiple test statistics vanishing simultaneously, requiring the use of multiple hypothesis testing procedures which we describe in \rref{subsec:implementation-details}.
Finally, in \rref{sec:empirical}, we demonstrate the performance of our algorithm in both a completely synthetic and a semi-synthetic setting.
In particular, we show that our method is capable of recovering the ground truth clustering with nearly 100\% accuracy even at relatively low sample sizes.
Our method also recovers the ground truth edges between observed nodes and latent nodes with higher accuracy than a baseline which does not make use of multiple hypothesis testing.
\subsection{Related Work}\label{sec:related}
\textbf{Learning undirected graphical models with clusters.} Since clusters of correlated variables are common across many disciplines, including biology \citep{eisen1998cluster}, economics \citep{bai2016econometric}, neuroscience \citep{arslan2018human,pircalabelu2020community}, and the behavioral sciences \citep{van2013handbook}, several structure learning methods have been developed which encourage clustering in the estimated graphs, especially in the setting of \textit{undirected} graphical models.
For example, \cite{tan2015cluster} introduced the \textit{cluster graphical lasso} method, which generalizes the traditional graphical lasso method to allow for the incorporation of known clustering information, resulting in denser estimated subgraphs over these clusters.
Building on this work, \cite{hosseini2016learning} introduce the \textit{GRAB} algorithm, which does \textit{not} require clusters to be known beforehand, but instead allows the clustering to be learned simultaneously to network structure.
More recently, \cite{pircalabelu2020community} introduced \textit{ComGGL}, a method which also learns clusters and graph structure simultaneously, with the additional benefit of high-dimensional consistency guarantees for both cluster recovery and graph structure in sparse settings.

\textbf{Latent tree models and factor analysis.} Unlike in the undirected settings above, in our setting, the clusters of observed variables are explicitly assumed to be induced by latent variables.
As has been observed in previous works, especially in latent tree modeling \citep{choi2011learning,shiers2016correlation,drton2017marginal,leung2018algebraic} and factor analysis \citep{drton2007algebraic,kummerfeld2016causal}, these latent variables produce ``signatures" or ``invariants" in the distribution over the observed variables, which can be exploited for structure learning.
One invariant which plays an important role in both settings is the \textit{tetrad} $t_{ij,uv}$, a $2\times 2$ subdeterminant of the correlation matrix which must \textit{vanish} (i.e., equal zero) whenever $i$ and $j$ share a single common latent parent, but have no children.
As we will see in \rref{sec:identifiability}, despite the differences in our model assumptions, tetrads also play an important role in our algorithm when identifying causal clusters.

\textbf{Traditional causal discovery methods.} Within the space of discovering causal models on observational data, there are two categorizations of algorithms. 
First, there are constraint-based methods, which rely on conditional independence testing to draw conclusions about the structure. 
The well-known PC-algorithm assumes \textit{causal sufficiency}, which bars unmeasured common cause latent variables and selection variables \citep{spirtes2000causation}. 
There also exist constraint-based methods on directed acyclic graphs with latent and selection variables, such as FCI, RFCI and their variants \citep{spirtes2001anytime, colombo2012learning}. 
These methods all learn Markov equivalence classes of directed acyclic graphs, as represented by completed partially directed acyclic graphs (CPDAGs) in the PC-algorithm or partial ancestral graphs (PAGs) in the FCI algorithm. 
The second categorization of methods is score-based algorithms, such as GES \citep{chickering2002optimal}, which identify underlying structure by optimizing a well-designed score function.
These methods, even those that are asymptotically correct in the presence of latent confounders, output equivalence classes of DAGs. In our work, we try to recover more complete causal information.

\textbf{Learning causal models with latent variables.}
Existing work for structure recovery in the presence of latent variables can often by characterized by the model structures that the method performs well or poorly upon. 
\cite{agrawal2021decamfounder} considers the model in which latent variables are \textit{pervasive}, influencing many observed nodes. 
Their method, \textit{DeCAMFounder}, recovers the true causal structure over observed variables by applying spectral decomposition in the non-linear additive noise and pervasive confounding setting \citep{agrawal2021decamfounder}, extending the linear setting of \cite{frot2019robust}.
A number of methods, similarly to the current work, also rely on algebraic constraints in the covariance matrix over observable variables to infer graph structure over latent variables.
The BPC \citep{silva2006learning} and FOFC \citep{kummerfeld2016causal} algorithms both leverage rank constraints on the covariance matrix to cluster observed variables, then recover some structure over the inferred latent nodes corresponding to these clusters.
Other algorithms, such as those proposed by \cite{shimizu2009estimation}, \cite{cai2019triad}, and \cite{xie2020generalized}, attempt to improve upon previous algorithms by restricting models to the \textit{linear non-Gaussian} case. 
In our work, we do not require non-Gaussianity, instead working in the general linear acyclic model regime. 
Furthermore, all of the above algorithms rely on the \textit{measurement assumption}, which requires that no observed variable is the parent of any latent variable.
This assumption, however, is not satisfied in many real-world applications of graphical models with latent variables, and in our work, we attempt to recover causal structures \textit{without} the measurement assumption. 

\section{Problem Setup}\label{sec:setup}
\begin{figure}
    \centering
    \includegraphics[width=\textwidth]{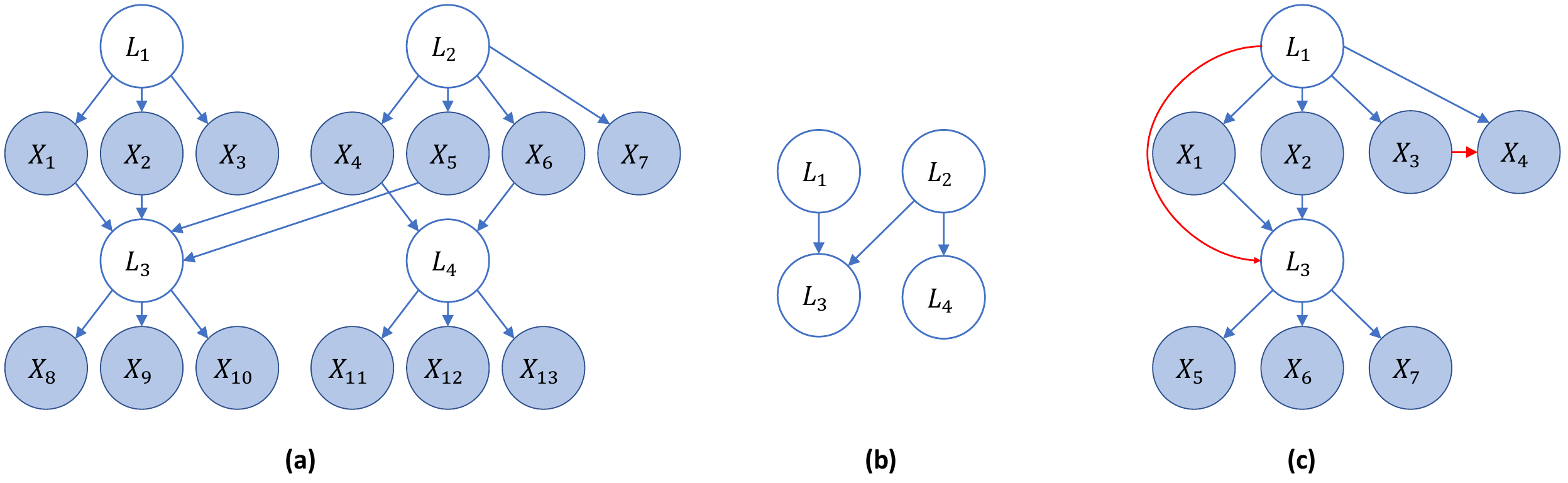}
    \caption{
    \textbf{(a)} $\mathcal{G}$ satisfies our model constraints.
    \textbf{(b)} The latent graph $L(\mathcal{G})$ for $\mathcal{G}$.
    \textbf{(c)} $\cG'$ falls outside of the class of models we consider in this paper, with violations shown in red.}
    \label{fig:model}
\end{figure}

We now formally define the class of models considered in this paper.
A \textit{structural causal model} (SCM) over the variables $\{ X_i \}_{i=1}^p$ consists of a set of \textit{structural assignments} of the form $X_i = f_i(X_{\pa(i)}, \epsilon_i)$, and a product distribution $\bbP_\epsilon$ over mean-zero \textit{exogenous noise} terms $\{ \epsilon_i \}_{i=1}^p$.
The set $\pa(X_i)$ are called the \textit{parents} of $X_i$, and the \textit{causal graph} for the SCM is a graph with nodes $\{ X_i \}_{i = 1}^p$ and directed edges $X_j \to X_i$ for $X_j \in \pa(X_i)$.
We assume that the causal graph for the SCM is \textit{acyclic}, in which case the distribution $\bbP_\epsilon$ induces a unique distribution $\bbP_X$ over $\{ X_i \}_{i=1}^p$.

In this paper, we focus on a class of SCMs with restrictions on both the structural assignments and on the causal graph.
First, we assume each $f_i$ is a \textit{linear} function, a common starting point for new methods, which has been the setting of many works \citep{chickering2002optimal,hauser2012characterization,solus2021consistency}.
Second, we assume that the causal graph is of the following form.
\begin{defn}\label{defn:latent-clustered-model}
Let $\cG$ be a DAG over latent nodes $L_1, \ldots, L_K$ and \emph{observed nodes} $X = \{X_1, \ldots, X_p\}$.
The \emph{clusters} of $\cG$ are the sets $C_k = \ch(L_k)$ for $k = 1, \ldots, K$.
The \emph{latent graph} for $\cG$, denoted $L(\cG)$, is the graph over $\{ L_k \}_{k=1}^K$ with an edge $k \to k'$ if and only if $X_k \to L_{k'}$ for some $X_k \in \ch(L_k)$.
$\cG$ is called a \emph{latent factor causal model} (LFCM) if it satisfies the following conditions:
\begin{itemize}
    \item[(a)][Unique cluster assumption] Each observed node has exactly one latent parent.
    
    \item[(b)][Bipartite assumption] There are no edges between pairs of observed nodes or between pairs of latent nodes.
    
    \item[(c)][Triple-child assumption] Each latent node has at least 3 observed children.
    
    \item[(c)][Double-parent assumption] If $k \to k'$ in $L(\cG)$, then there exist two nodes $X_i, X_j \in \ch(L_k)$ such that $X_i \to L_{k'}$ and $X_j \to L_{k'}$.
\end{itemize}
\end{defn}

See \rref{fig:model}a~for an example of a graph that satisfies our model definition, and \rref{fig:model}c~for a graph that does not.
The importance of each assumption for the purpose of identifying $\cG$ will become clear in the proofs of the genericity and identifiability results presented in the next section.
For example, we will see that the edge $X_3 \to X_4$ in \rref{fig:model}c prevents the submatrix $\Sigma_{[1,2],[3,4]}$ of the covariance matrix $\Sigma$ from being low rank, and thus prevents $\{ X_1, X_2 \}$ and $\{ X_3, X_4 \}$ from being clustered.

\section{Trek separation and genericity assumptions}\label{sec:identifiability}
In this section, we review fundamental results essential to the identifiability of LFCMs, which we constructively prove in \rref{sec:methods} by introducing an algorithm for consistently estimating LFCMs.
We will also introduce genericity assumptions necessary for the consistency of our algorithm.

We denote the covariance matrix of our model as $\Sigma$, and given two subsets of nodes $A, B$, we use $\Sigma_{A,B}$ to denote the submatrix of $\Sigma$ with rows in $A$ and columns in $B$.
Our identifiability results rely on a common generalization of d-separation, known as \textit{trek separation}, which relates the causal graph of a SCM to the rank of submatrices of $\Sigma$.
A \textit{directed path} from node $i$ to node $j$ is a sequence of nodes $p_1 = i, ..., p_k = j$, such that $p_i \to p_{i+1}$ for all $i$ from $1$ to $k-1$.
In this case, $j$ is called the \textit{sink} of the path and $i$ is called the \textit{source}.
A \textit{trek} in the graph $\mathcal{G}$ from $i$ to $j$ is an ordered pair of directed paths $(P_1, P_2)$, such that the sink of $P_1$ is $i$ and the sink of $P_2$ is $j,$ and $P_1, P_2$ share a source $k.$
%
%
Now, we define \textit{trek separation}. 
Given four subsets $A, B, C_A, C_B$ of nodes (note these subsets need not be disjoint), the pair $(C_A, C_B)$ \textit{t-separates} $A$ and $B$ if, for every trek $(P_1, P_2)$ between $A$ and $B,$ $P_1$ contains a node in $C_A$ or $P_2$ contains a node from $C_B.$ 
Finally, the following theorem relates the notion of \textit{t-separation} to the rank of submatrices of the covariance matrix.
\begin{thm}[Trek separation, \cite{sullivant2010trek}]\label{thm:trek-separation}
Let $A, B$ be two subset of nodes in $\cG$.
Then
$$
\rank(\Sigma_{A,B}) \leq \min \{|C_A| + |C_B|: (C_A, C_B)~\text{t-separates }A \text{ from }B \text{ in } \cG \}
$$
Moreover, equality holds generically\footnote{We say a statement holds \textit{generically} if the set of parameters for which it does not hold has Lebesgue measure zero.} for $\Sigma$ consistent with $\cG$.
\end{thm}

In this paper, we only need to use information about the rank of 2 x 2 submatrices of $\Sigma$.
The determinants of these matrices are commonly known as \textit{tetrads}.
In particular, we denote $t_{ij,uv} = \det(\Sigma_{[ij] , [uv]}) = \Sigma_{i u} \Sigma_{j v} - \Sigma_{i v} \Sigma_{j u}$.
Specializing \rref{thm:trek-separation}, we obtain the following corollary:
\begin{cor}[Tetrad representation, \cite{spirtes2013calculation}]\label{corollary:tetrad-representation}
    Suppose $A = \{ X_i, X_j \}$ and $B = \{ X_u, X_v \}$ are t-separated by a single node.
    Then $t_{ij,uv} = 0$.
\end{cor}

We can now see the importance of the first three assumptions in \rref{defn:latent-clustered-model}.
These structural assumptions control the size of t-separating sets between nodes in the same cluster and in different clusters, so that we can apply \rref{thm:consistency} and \rref{corollary:tetrad-representation} to ensure that certain tetrads are either zero or generically non-zero.
In particular, the unique cluster assumption and the bipartite assumption guarantees that two nodes $X_i$ and $X_j$ in the same cluster will be t-separated from their non-descendants by their latent parent.
Thus, clusters of nodes with no descendants can be identified.
Conversely, the triple-child assumption ensures that two nodes $X_i$ and $X_j$ that are not in the same cluster do not get clustered, since we can find a 2x2 submatrix with $i$ indexing one of the rows and $j$ indexing one of the columns that is generically of rank 2.
In \rref{appendix:faithfulness-generic}, we formally state and prove that the following faithfulness assumptions are indeed generic under the first 3 structural assumptions from \rref{defn:latent-clustered-model}:

\begin{restatable}[Cluster tetrad faithfulness]{assumption}{clusterFaithfulness}\label{assumption:cluster-tetrad-faithfulness}
Suppose $X_i$ and $X_j$ are not in the same cluster.
Then there exists some $\{ u, v \}$ such that $t_{ij,uv} \neq 0$.
\end{restatable}

\begin{restatable}[Parent tetrad faithfulness]{assumption}{parentTetradFaithfulness}
\label{assumption:parent-tetrad-faithfulness}
Suppose $X_i$ and $X_j$ are in the same cluster, but $X_i$ has at least one child.
Then there exists some $\{ u, v \}$ such that $t_{ij,uv} \neq 0$.
\end{restatable}

\begin{restatable}[Latent adjacency faithfulness]{assumption}{latentAdjacencyFaithfulness}
\label{assumption:latent-adjacency-faithfulness}
Suppose $X_i \to L_k$.
Let $S_i = \ch(\pa(X_i)) \setminus \{ i \}$ and $S' = \cup_{j \leq i} \ch(L_i)$.
Then $\rho_{i, k \mid S_i, S} \neq 0$ for some $X_k \in \ch(L_k)$
\end{restatable}

\begin{remark}
Since causal structure learning algorithms are always run in a noisy setting, \emph{near} violations of genericity assumptions can degrade the performance of a method, as discussed by \cite{uhler2013geometry}.
In particular, the set of parameters which violate a ``strong" faithfulness condition is generally a positive measure set, extending from the measure zero set where faithfulness is violated.
Fortunately for the current setting, our assumptions require the existence of only a \emph{single} entry of the underlying statistic being far from zero.
Thus, the set of parameters violating the ``strong" version of our faithfulness assumption is an \emph{intersection} of the sets of parameters for which each $t_{ij,uv}$ is near zero, resulting in a smaller set.
In the present work, we will not attempt to quantify the size of this set and the resulting statistical benefits, but note these as interesting directions for future work.
\end{remark}

\section{Methods}\label{sec:methods}

\begin{wrapfigure}{R}{.43\linewidth}
\begin{minipage}{\linewidth}
\vspace{-0.8cm}
\input{contents/alg-main}
\vspace{-0.8cm}
\end{minipage}
\end{wrapfigure}

Our algorithm, presented in \rref{alg:main}, consists of three stages.
As is common in causal structure learning, we present our algorithm with implementation details of hypothesis testing abstracted away.
In particular, we will assume access to two subroutines, whose implementation details will be given in \rref{subsec:implementation-details}.
The first subroutine tests $H_{ci}(X_j, X_A \mid X_B)$, which denotes the null hypothesis that $X_j$ and $X_A$ are conditionally independent given $X_B$.
The second subroutine tests $H_{vt}(X_A, X_B)$, which denotes the null hypothesis that all tetrads of $\Sigma_{A, B}$ vanish.

In the \textbf{first stage} (\rref{alg:find-bottom}, see also \rref{fig:phase1}), we identify clusters of observed variables with the same latent parent.
However, note that since this stage only identifies \textit{leaves} with the same latent parent, it is not guaranteed to identify all nodes with the same latent parent.
This stage simultaneously recovers an ordering over these clusters.
Thus, in the \textbf{second stage} (\rref{alg:mergeCluster}, see also \rref{fig:phase23}a), we iterate over pairs of clusters output from the first stage, identify pairs of clusters with the same latent parent, and merge them, while leaving the ordering of the clusters intact.
In the \textbf{third stage} (\rref{alg:learnDAG}, see also \rref{fig:phase23}bc), we use the clustering and ordering information discovered in the previous two stages to learn a DAG over both latent and observed variables.
In particular, given a node $X_j$ in cluster $C_j$ which comes before the cluster $C_i$ in our ordering, we wish to determine whether $X_j$ has an edge to the associated latent variable $L_i$.
By \rref{assumption:parent-tetrad-faithfulness}, this can be accomplished by checking partial correlations between $X_j$ and the nodes in $C_i$.
These stages compose a consistent algorithm, as established in the following theorem and proven in \rref{appendix:algorithm-consistency}.
In \rref{sec:identifiability}, we have already discussed the importance of the first three structural assumptions from \rref{defn:latent-clustered-model}.
In \rref{appendix:double-parent-violation}, we show how our algorithm fails under a violation of the double-parent assumption.

\begin{figure}
    \centering
    \includegraphics[width = \textwidth]{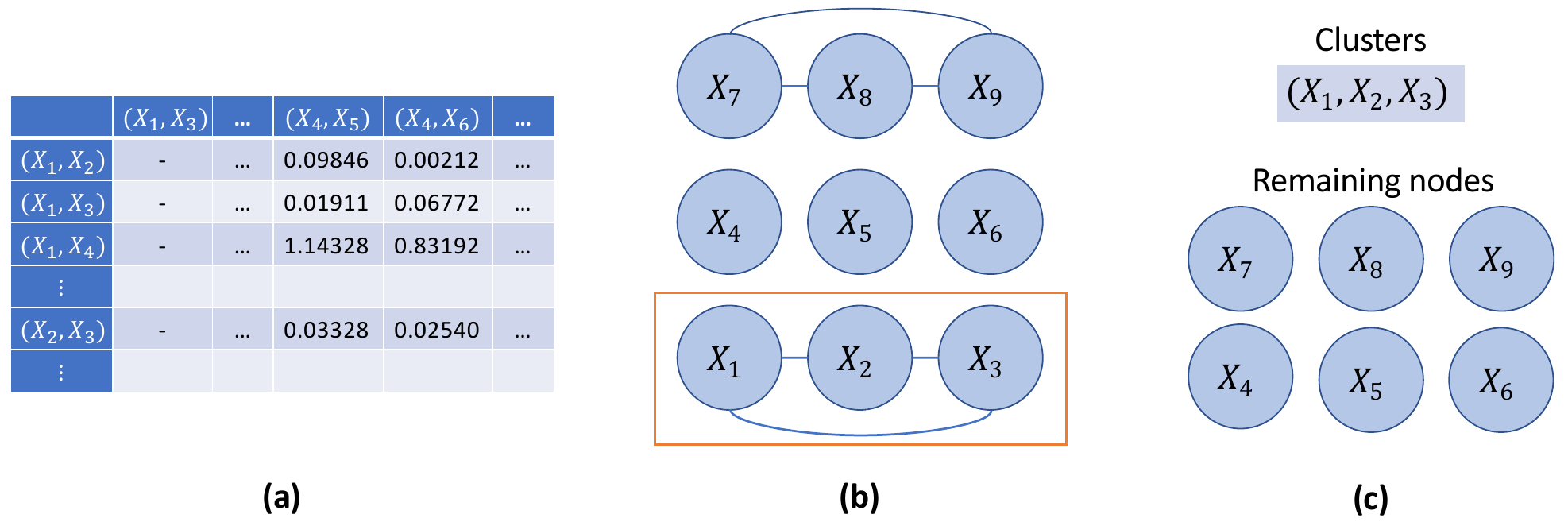}
    \caption{\textbf{Phase 1}: 
    Let the true LFCM be the graph in \rref{fig:phase23}(c).
    In our algorithm's first phase (see \rref{alg:find-bottom}) we perform the following steps:
    \textbf{(a)} Compute tetrad scores between all pairs of nodes.
    \textbf{(b)} For each pair of nodes, test the null hypothesis that all tetrads are zero. Construct a graph with the edge $i - j$ for any pair of nodes where we do not reject the null hypothesis.
    \textbf{(c)} Extract a clique from this graph to be a cluster (e.g., by picking the largest clique with arbitrary tie breaking), remove these nodes and repeat with remaining nodes.
    }
    \label{fig:phase1}
\end{figure}

\begin{restatable}{thm}{consistency}
\label{thm:consistency}
Let $\cG$ be a linear LFCM and let $\bbX \in \bbR^{n \times p}$ be a matrix of samples of the observed variables $X_1, \ldots, X_p$.
Then \rref{alg:main} is consistent under Assumptions \ref{assumption:cluster-tetrad-faithfulness}, \ref{assumption:parent-tetrad-faithfulness}, and \ref{assumption:latent-adjacency-faithfulness}, i.e., as $n \to \infty$, we have $\bbP(\hat{\cG} = \cG) \to 1$.
\end{restatable}

Next, we outline the complexity of our algorithm, using placeholders for the complexities of hypothesis tests in order to keep our results general.
Let $f(d, n)$ denote the cost of performing the hypothesis test $H_{vt}$ on $d$ statistics from $n$ samples, when all sufficient statistics are pre-computed.
Let $g(p)$ denote the complexity of an algorithm used to find a clique in a graph on $p$ nodes.
While finding the \textit{largest} clique in a graph is in general NP-hard, we can avoid this complexity, since the consistency of our algorithm does not rely on picking the largest clique at each step, only a clique of 3 or more nodes (picking larger cliques is simply a tool for improving statistical accuracy).

In the following, let $M$ be the maximum size of any returned cluster, and let $K$ be the number of clusters discovered by the algorithm.
By definition, $M \leq p$ and $K \leq p$, so replacing these quantities by $p$ gives complexities that are only in terms of the known problem parameters.
However, such upper bounds can be highly pessimistic.
If the true graph has few nodes per cluster, or a small number of latent nodes, then with enough samples, $M$ and $K$ will also be small, respectively.

\begin{figure}[t]
\begin{minipage}{\textwidth}
    \begin{minipage}[t]{.46\textwidth}
    \input{contents/alg-phase1}
    \end{minipage}
    ~~
    \begin{minipage}[t]{.48\textwidth}
    \input{contents/alg-phase2}
    \vspace{-0.7cm}
    \input{contents/alg-phase3}
    \end{minipage}
\end{minipage}
\end{figure}

\begin{thm}
The complexity of each algorithm is:
\begin{itemize}
    \item[(a)] \rref{alg:find-bottom} takes $\cO(p^4 + p^3 f(p^2, n) + p g(p))$.
    
    \item[(b)] \rref{alg:mergeCluster} takes $\cO(p^2 M^4)$.
    
    \item[(c)] \rref{alg:learnDAG} takes $\cO(p K M)$.
\end{itemize}
\end{thm}
\begin{proof}
\textbf{(a)} Computing all tetrads and their associated p-values is $\cO(p^4)$.
In each round, we perform $\cO(p^2)$ hypothesis tests, each on $\cO(p^2)$ statistics, so that the complexity at each round is $\cO(p^2 f(p^2, n))$.
After performing these tests, we identify the largest clique in a graph of $\cO(p)$ nodes, so that the total run time per round is $\cO(p^2 f(p^2, n) + g(p))$.
At most $p$ rounds are required, resulting in the stated complexity.

\textbf{(b)} We must check $\cO(p^2)$ pairs of clusters for whether or not they should be merged, and the maximum size of the union of any such pair is $\cO(M)$.
To check whether $\cO(M)$ nodes belong to the same cluster, we require a hypothesis test on $\cO(M^4)$ statistics, so that this step takes $\cO(p^2 f(M^4, n))$.

\textbf{(c)} We perform $\cO(p K)$ hypothesis tests, each based on $\cO(M)$ partial correlations.
\end{proof}

Assume that we use the Sidak adjustment procedure explained in the next section (so that $f(d, n) = \cO(d)$), and a greedy algorithm for picking cliques (so that $g(p) = p^3$).
Then, replacing $K$ and $M$ by $p$, we have that \rref{alg:find-bottom} takes $\cO(p^5)$, \rref{alg:mergeCluster} takes $\cO(p^6)$, and \rref{alg:learnDAG} takes $\cO(p^3)$, so that the overall complexity of our algorithm is at most $\cO(p^6)$.
Even in this pessimistic analysis, this complexity is relatively low for causal structure learning, which is known to be NP-hard in general \citep{chickering2004large}, and for which variants of the best-known algorithms, such as PC \citep{spirtes2000causation} and GES \citep{chickering2002optimal}, typically have complexity $\cO(p^{d+2})$, where $d$ is the maximum in-degree of the graph \citep{chickering2020statistically}.

\subsection{Implementation Details}\label{subsec:implementation-details}

The null hypothesis $H_{vt}$ and $H_{ci}$ used in Algorithms \ref{alg:find-bottom}, \ref{alg:mergeCluster}, and \ref{alg:learnDAG} imply that some vector-valued statistic of the covariance matrix is equal to zero.
Procedures for \textit{simultaneous hypothesis testing} are designed to (asymptotically) control the false discovery rate (FDR) of such a form of hypothesis test.
In practice, we found that computing marginal p-values and performing \textit{Sidak adjustment} \cite{drton2007multiple} yields good performance.
In particular, given p-values $\{ \pi_m \}_{m=1}^M$, the Sidak-adjusted p-values are
\[
\pi_m^\sidak = 1 - (1 - \pi_m)^M
\]

\begin{remark}
    The Sidak adjustment uses only the \textit{marginal} distributions of each tetrad, neglecting potentially important information about the \textit{correlations} between tetrads.
    In contrast, the \emph{max-T adjustment} accounts for correlations between the tested statistics by estimating their correlation matrix, and has been shown to outperform the Sidak adjustment both theoretically and in practice \citep{drton2007multiple,chernozhukov2013gaussian}.
    However, the max-T adjustment requires sampling from a potentially high-dimensional multivariate normal distribution, an operation which is $\cO(d^3)$ for dimension $d$.
    We have found that in practice, max-T adjustment performs similarly to Sidak adjustment while taking substantially longer.
    Therefore, we use Sidak adjustment for our experimental results, but provide capability for max-T adjustment in our codebase.
\end{remark}

Given a set adjusted p-values and a significance level $\alpha$, we reject the null hypothesis if \textit{any} of the adjusted p-values are smaller than $\alpha$.
To test conditional independence, recall that $H_{ci}(X_j, X_A \mid X_B)$ holds in a multivariate normal if and only if the vector of partial correlations $\{ \rho_{ij \mid B} \}_{i \in A}$ is zero.
To compute p-values, we use a widely used procedure which we call the \textit{Fisher correlation test}.
First, given the sample partial correlations $\{ \hat{\rho}_{ij \mid B} \}_{i \in A}$, we apply the \textit{Fisher z-transformation} $\hat{z}_{ij \mid B} = \sqrt{n - |B| - 3} \arctanh(\hat{\rho}_{ij \mid B})$.
Then, we compute the two-tailed p-value of $\hat{z}_{ij \mid B}$ with respect to $\cN(0, 1)$, i.e., $\pi_{ij\mid B} = 2 Q(|\hat{z}_{ij \mid B}|)$, where $Q$ is the tail distribution function of $\cN(0, 1)$.

Next, to test $H_{vt}(X_A, X_B)$, we adopt the widely-used \textit{Wishart test} to compute the p-values \citep{wishart1928sampling,kummerfeld2016causal}, which we now briefly describe.
First, we compute the \textit{sample tetrads} $\hat{t}_{ij,uv} = \hSigma_{iu} \hSigma_{jv} - \hSigma_{iv} \hSigma_{ju}$ for $\{i, j \} \subset A$ and $\{ u, v \} \subset B$ such that $i, j, u$ and $v$ are distinct.
Then, we normalize each sample tetrad, dividing by an estimate of its standard deviation to obtain the z-score $\hat{z}_{ij,uv}$.
\cite{drton2008moments} give the following formula for the variance of sample tetrads in terms of the true covariance matrix $\Sigma$:
\[
\Var\left( \hat{t}_{ij,uv} \right) = n \cdot (n - 1)^{-3} \cdot \left( (n + 2) |\Sigma_{[ij],[ij]}| \cdot |\Sigma_{[uv],[uv]}| - n |\Sigma_{[ijuv],[ijuv]}| + 3 n |\Sigma_{[ij],[uv]}| \right),
\]
where $|A| = \det(A)$. 
To estimate the variance, we use the above formula with the sample covariance $\hSigma$ replacing $\Sigma$. 
Finally, we compute the two-tailed p-value of $\hat{z}_{ij,uv}$ with respect to $\cN(0, 1)$.
\begin{figure}
    \centering
    \includegraphics[width = \textwidth]{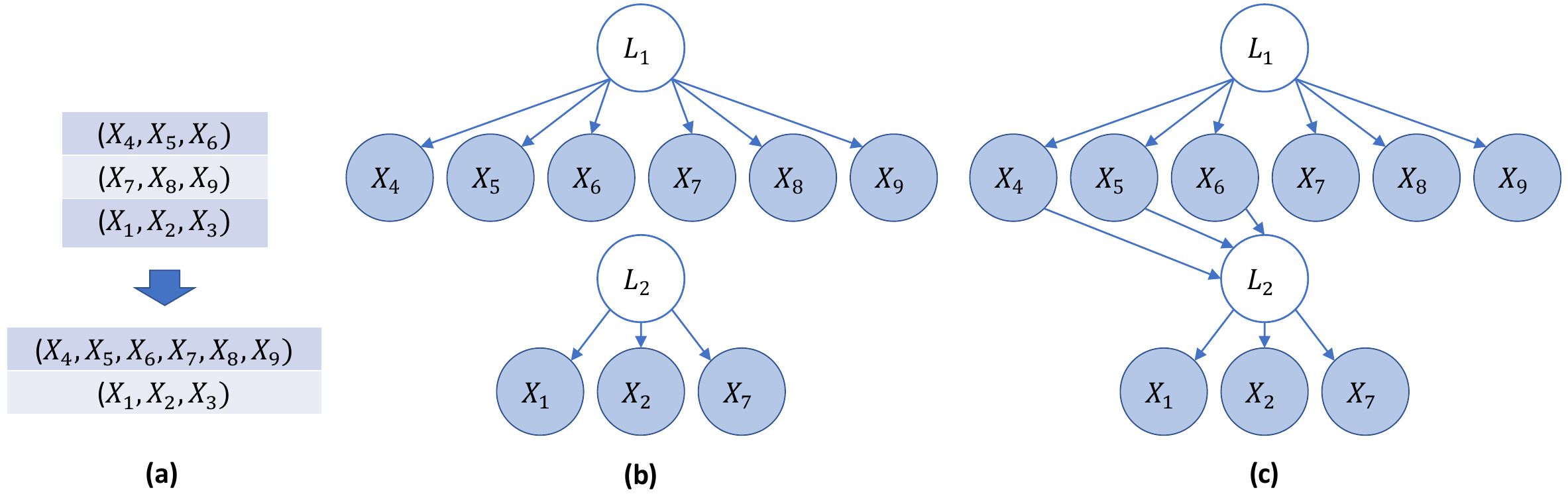}
    \caption{\textbf{Phases 2 and 3.} 
    \textbf{(a)} Merge pairs of clusters based on vanishing tetrad tests.
    \textbf{(b)} Introduce latent nodes, and add edges from latent nodes to children.
    \textbf{(c)} Add parents of latent nodes based on conditional independence testing.
    }
    \label{fig:phase23}
\end{figure}

\section{Empirical Results}\label{sec:empirical}

We evaluate our algorithm in two settings.
First, we evaluate in a purely synthetic setting, which allows us to generate SCMs which exactly match our proposed model.
Then, we evaluate in a semi-synthetic setting, modifying real data to more closely match our proposed model while demonstrating that our approach has promise in real-world biological settings.

\subsection{Synthetic data experiments}

We begin by briefly describing the simulation settings used for our experiments, before describing the baselines and metrics which we use for evaluation.
We generate a graph with 10 latent nodes, we first sample a ``latent" skeleton $L(\cG)$ over $\{ 1, 2, \ldots, 10 \}$ from a directed Erd{\"o}s-R{\'e}nyi model with edge probability $0.5$.
Then, for each latent node $L_k$, we generate $c_k \sim \Unif(3, 6)$ children.
Finally, for each edge $L_k \to L_{k'}$ in $L(\cG)$, we sample $d_{k,k'} \sim \Unif(2, |\ch(L_k)|)$, then sample $d_{k,k'}$ children of $L_k$. 
For each selected child $c_k$, we add the edge $c_k \to L_{k'}$, giving us a DAG $\cG$ over both latent and observed nodes which has latent skeleton $L(\cG)$ and satisfies \rref{defn:latent-clustered-model}.

Given this DAG, we generate a linear SCM as follows, proceeding in topological order.
For each node $X_j$, if the node has no parents, its equation is $X_j = \epsilon_j$ for $\epsilon_j \sim \cN(0, 1)$.
If the node has parents, then for each parent $X_i$, we sample an ``initial" weight $\tilde{w}_{ij} \sim \Unif([-1, -.25] \cup [.25, 1])$.
Next, we describe how to normalize these weights in order to avoid the \textit{varsortability} issue described by \cite{reisach2021beware}, where simulated DAGs are easy to learn because the variance of each node tends to increase according to the topological order.
Given these initial weights, we simulate $B$ ``parental contributions" $\mu_j^{(b)} = \sum_{i \in \pa_\cG(j)} \tilde{w}_{ij} X_i^{(b)}$ for $X^{(b)}$ sampled from the linear SEM defined over $i < j$.
The sample variance $\hsigma_j$ of $\mu_j^{(b)}$ serves as an estimate for the variance that the parents of $j$ will contribute to $X_j$.
Finally, we ensure that $X_j$ has variance 1 and that half of its variance is contributed by its parents by setting the final weights as $w_{ij} = (2 \hsigma_j)^{-1/2} \tilde{w}_{ij}$ and $\epsilon_j \sim \cN(0, 1/2)$.

\textbf{Accuracy of learning clusters.} In our first set of experiments, we evaluate the accuracy of the learned clusters.
To measure the accuracy over the learned clustering compared to the underlying clustering, we use the following criteria: the pair $(X_i, X_j)$ is a \textit{true positive} if $X_i$ and $X_j$ are in the same underlying cluster and are in the same learned cluster, the pair is a \textit{false positive} if $X_i$ and $X_j$ are not in the same underlying cluster but are in the same learned cluster, and so on.
We generate 50 different SEMs via the process described above, and from each SEM we generate $n = 200$ samples.
We run our algorithm using significance levels ranging from $.05$ and $.5$.
The results are shown in \rref{fig:phase1_results}. Due to interactions between the hypothesis tests used by our algorithm (denoted ``LFCM", shown in blue), the ROC curve is highly non-monotonic over larger ranges of values, so that the curve occupies only a small range of the plot, though it clearly drastically outperforms the competing methods, achieving almost perfect performance.
In particular, we consider two baselines.
First, we compare to \textit{spectral clustering} (denoted ``SC", shown in green), as implemented in \texttt{sklearn} in Python, a widely used clustering technique in genomics \citep{higham2007spectral}, with a varying number of estimated clusters from 2 to 30.
Second, we compare to the FindOneFactorCluster algorithm of \cite{kummerfeld2016causal} (denoted ``FOFC", shown in red).
We found that spectral clustering performs slightly better than random guessing, while FOFC performs about the same as random guessing, reflecting the drastic deviation from the measurement assumption on which it relies.
Finally, we verified that randomly picking $K$ clusters of equal size (denoted ``Random", shown in orange), for $K$ varying from 2 to 30, matched the diagonal random guessing line.

\begin{figure}
    \centering
    \begin{subfigure}{.48\linewidth}
    \includegraphics[width=\textwidth]{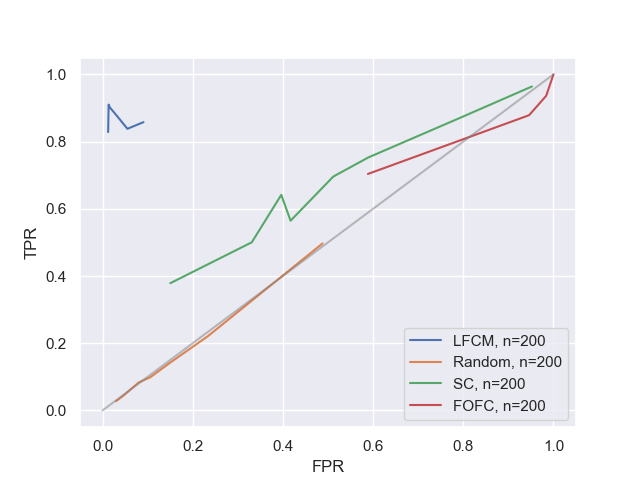}
    \caption{\textbf{Cluster recovery performance.}}
    \label{fig:phase1_results}
    \end{subfigure}
    \begin{subfigure}{.48\linewidth}
    \includegraphics[width=\textwidth]{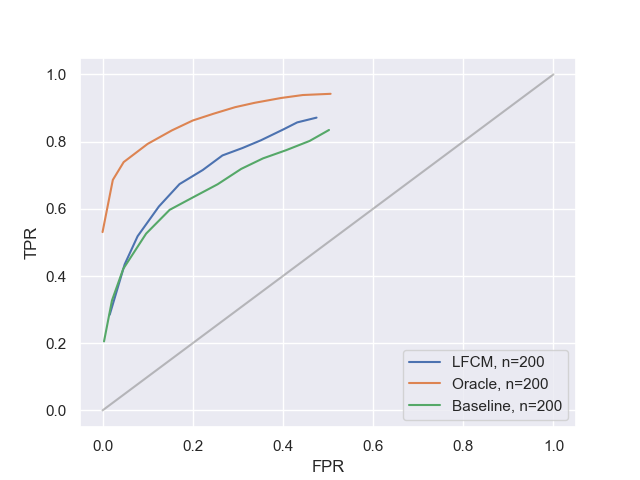}
    \caption{\textbf{Edge recovery performance.}}
    \label{fig:phase2_results}
    \end{subfigure}
    \caption{\textbf{Performance on synthetic data.} The first two phases of our algorithm almost perfectly recover the ground truth clusters, while the third phase of our algorithm demonstrates the utility of multiple hypothesis testing for recovering edges between observed nodes and latent nodes.}
\end{figure}

\textbf{Accuracy of learning edges from observed nodes to latent nodes.} In our second set of experiments, we evaluate the accuracy of learning the edges from observed nodes to latent nodes, when the true clusters and their ordering is known.
In particular, for $L_i \prec L_j$ in the ordering, and $X_i \in \ch(L_i)$, the pair $(X_i, L_j)$ is considered a \textit{true positive} if $X_i \to L_j$ in the true LFCM as well as in the estimated LFCM, a \textit{false positive} if it is not in the true LFCM but does appear in the estimated LFCM, and so on.
In \rref{fig:phase2_results}, we compare the third phase of our algorithm (denoted ``LFCM", shown in blue) to a \textit{baseline} which simply uses a single child of each latent node for the conditional independence test (denoted ``Baseline", shown in green), as well as an \textit{oracle} which is able to observe the values of the latent nodes and is thus infeasible (denoted ``Oracle", shown in orange).
As expected, our algorithm does not perform as well as this unrealizable case, but still performs significantly better than random (the diagonal line) and noticeably better than the baseline.

\subsection{Semi-synthetic experiments on protein signaling data}

In this section, we demonstrate the applicability of our method to a real-world dataset in a semi-synthetic setting.
The Sachs protein mass spectroscopy dataset \citep{sachs2005causal} is a widely used benchmark for causal discovery, in part due to the existence of a commonly accepted ground truth network over the 11 measured protein expression values, shown in \rref{fig:real-data}a.
We use the 1,755 ``observational" samples, where the experimental conditions involve only perturbing receptor enzymes, and not any signaling molecules, as described in \cite{wang2017permutation}.
To make the ground truth network more similar to a latent factor causal model, we perform three data-processing steps: (1) we ``condition" on PKA, by regressing it out of the dataset, (2) we ``remove" the direct effect of Raf on Mek, and (3) we ``marginalize" out PIP3 and PKC by removing the corresponding columns from the dataset.
We ``remove" the direct effect of Raf on Mek as follows.
First, we regress Mek on its two remaining parents, Raf and PKC.
Call the resulting regression coefficient for Raf $\beta_{Raf}$.
For each sample, we subtract the value of Raf times the $\beta_{Raf}$ from the value of Mek.
Note that we do \textit{not} remove the direct effect of PLC$\gamma$ on PIP2, since then our algorithm collapses all nodes into a single cluster.
The processed graph is show in \rref{fig:real-data}b.

Running our method with significance level $\alpha = 0.01$ for $H_{vt}$ and $\alpha = 0.1$ for $H_{ci}$, we obtain the network shown in \rref{fig:real-data}c.
The clustering by our algorithm closely matches the clustering (\textit{Akt, PLC$\gamma$, PIP2}), (\textit{p38, JNK, Raf, Mek, Erk}) induced by the true network, with the exception that \textit{Akt} from the first cluster and \textit{Erk} from the second cluster are pulled out into a cluster with one another, which may indicate that the effect of PKA on Akt and Erk cannot be completely removed using a purely linear approach.
The ordering between the clusters (\textit{PLC$\gamma$, PIP2}) and  (\textit{p38, JNK, Raf, Mek}) is preserved, but the edge $PIP2 \to L3$ is missing.

\begin{figure}
    \centering
    \includegraphics[width=\textwidth]{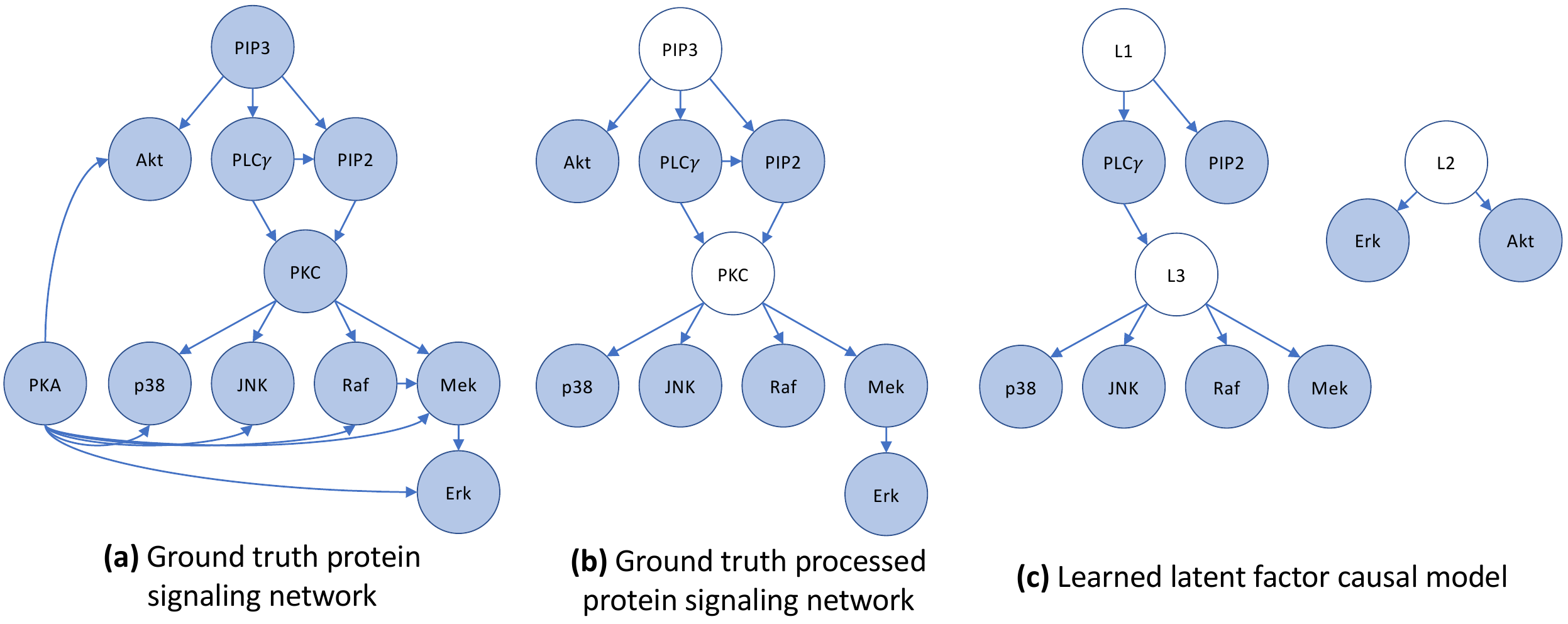}
    \caption{\textbf{Learning a latent factor causal model for protein signaling}. Our recovered model in \textbf{(c)} nearly captures the ground truth network in \textbf{(b)}.}
    \label{fig:real-data}
\end{figure}


\section{Discussion}\label{sec:discussion}

In this paper, we introduce a method (\rref{alg:main}) for learning \textit{latent factor causal models} (LFCMs), a novel, biologically-motivated class of causal models with latent variables.
We showed that these models are identifiable in the linear setting using rank constraints on submatrices of the covariance matrix, and that our method provides a consistent estimator for these models.
We also showed that our method outperforms existing clustering algorithms on synthetic data, and almost perfectly recovers a widely-accepted ground truth network in a semi-synthetic biological setting.
These results serve as a proof-of-concept, suggesting that our algorithm may be able to shed biological insight on the problem of identifying the spatial clustering of genes in the cell nucleus given data on the expression of the genes. 
Interestingly, since it is possible (although expensive) to measure the 3D organization of the genome in the cell nucleus ~\citep{Lieberman}, there is a meaningful avenue to validate our method on the important biological application of connecting 3D genome organization with gene expression~\citep{Shiva_review}.
We conclude with a discussion of the limitations of our model, which suggests a number of other directions for future work.

\textbf{Limitations.} The latent factor causal model (LFCM) class considered in this paper has two obvious limitations.
First, we make the strong parametric assumption of a linear Gaussian SEM.
While many nonparametric conditional independence tests have been proposed \citep{gretton2007kernel,zhang2011kernel}, we are not aware of nonparametric tests for shared latent factors that would generalize $H_{vt}$.
Thus, extending our algorithm to a nonparametric setting would require development of such tests.
In particular, generalizing tetrad constraints to the nonlinear setting is an interesting direction for future research.
Second, we make two strong structural assumptions.
The ``unique cluster" assumption is well-motivated by our biological setting of interest, and is likely the easiest assumption to remove since \rref{thm:trek-separation} already provides a generalization of the rank constraint we leverage.
Indeed, generalized rank constraints have already been explored in prior work on factor analysis \citep{drton2007algebraic,kummerfeld2014causal}.
The ``bipartite assumption" has two components which may be separately examined.
First, the assumption that there are no edges between observed variables is most reasonable for systems such as gene regulatory networks where a different, unobserved entity class (in this case, proteins) mediates \textit{all} interactions between the observed variables (i.e., genes).
This assumption may also be expendable, for instance by allowing for a \textit{small number} of edges between observed variables, akin to the low-rank plus sparse literature in previous work on learning with exogenous latent variables \citep{frot2019robust,agrawal2021decamfounder}.
Similarly, existing techniques \citep{cai2019triad,xie2020generalized} may help to eliminate the assumption that there are no edges between latent variables.
In addition to relaxing these assumptions, it would be of interest to develop procedures for testing these assumptions in data, e.g., by extending recent work \citep{agarwal2020synthetic} which develops a spectral-energy-based hypothesis test for structural assumptions in latent variable models.


\acks{
Chandler Squires was partially supported by  an  NSF  Graduate  Research  Fellowship.
All authors were partially supported by NSF (DMS-1651995), ONR (N00014-17-1-2147 and N00014-22-1-2116), the MIT-IBM Watson AI Lab, MIT J-Clinic for Machine Learning and Health, the Eric and Wendy Schmidt Center at the Broad Institute, and a Simons Investigator Award to Caroline Uhler.
}

\bibliography{bib}

\begin{thebibliography}{55}
\providecommand{\natexlab}[1]{#1}
\providecommand{\url}[1]{\texttt{#1}}
\expandafter\ifx\csname urlstyle\endcsname\relax
  \providecommand{\doi}[1]{doi: #1}\else
  \providecommand{\doi}{doi: \begingroup \urlstyle{rm}\Url}\fi

\bibitem[Agarwal et~al.(2020)Agarwal, Shah, and Shen]{agarwal2020synthetic}
Anish Agarwal, Devavrat Shah, and Dennis Shen.
\newblock Synthetic interventions.
\newblock \emph{arXiv preprint arXiv:2006.07691}, 2020.

\bibitem[Agrawal et~al.(2021)Agrawal, Squires, Prasad, and
  Uhler]{agrawal2021decamfounder}
Raj Agrawal, Chandler Squires, Neha Prasad, and Caroline Uhler.
\newblock The decamfounder: Non-linear causal discovery in the presence of
  hidden variables.
\newblock \emph{arXiv preprint arXiv:2102.07921}, 2021.

\bibitem[Arslan et~al.(2018)Arslan, Ktena, Makropoulos, Robinson, Rueckert, and
  Parisot]{arslan2018human}
Salim Arslan, Sofia~Ira Ktena, Antonios Makropoulos, Emma~C Robinson, Daniel
  Rueckert, and Sarah Parisot.
\newblock Human brain mapping: A systematic comparison of parcellation methods
  for the human cerebral cortex.
\newblock \emph{NeuroImage}, 170:\penalty0 5--30, 2018.

\bibitem[Bai and Wang(2016)]{bai2016econometric}
Jushan Bai and Peng Wang.
\newblock Econometric analysis of large factor models.
\newblock \emph{Annual Review of Economics}, 8:\penalty0 53--80, 2016.

\bibitem[Belyaeva et~al.(2017)Belyaeva, Venkatachalapathy, Nagarajan,
  Shivashankar, and Uhler]{Belyaeva}
Anastasiya Belyaeva, Saradha Venkatachalapathy, Mallika Nagarajan, G.V.
  Shivashankar, and Caroline Uhler.
\newblock Network analysis identifies chromosome intermingling regions as
  regulatory hotspots for transcription.
\newblock \emph{PNAS}, 114\penalty0 (52):\penalty0 13714--13719, 2017.

\bibitem[Bernstein et~al.(2020)Bernstein, Saeed, Squires, and
  Uhler]{bernstein2020ordering}
Daniel Bernstein, Basil Saeed, Chandler Squires, and Caroline Uhler.
\newblock Ordering-based causal structure learning in the presence of latent
  variables.
\newblock In \emph{International Conference on Artificial Intelligence and
  Statistics}, pages 4098--4108. PMLR, 2020.

\bibitem[Blalock(2017)]{blalock2017causal}
Jr~Blalock.
\newblock \emph{Causal models in the social sciences}.
\newblock Routledge, 2017.

\bibitem[Cai et~al.(2019)Cai, Xie, Glymour, Hao, and Zhang]{cai2019triad}
Ruichu Cai, Feng Xie, Clark Glymour, Zhifeng Hao, and Kun Zhang.
\newblock Triad constraints for learning causal structure of latent variables.
\newblock \emph{Advances in Neural Information Processing Systems},
  32:\penalty0 12883--12892, 2019.

\bibitem[Chernozhukov et~al.(2013)Chernozhukov, Chetverikov, and
  Kato]{chernozhukov2013gaussian}
Victor Chernozhukov, Denis Chetverikov, and Kengo Kato.
\newblock Gaussian approximations and multiplier bootstrap for maxima of sums
  of high-dimensional random vectors.
\newblock \emph{The Annals of Statistics}, 41\penalty0 (6):\penalty0
  2786--2819, 2013.

\bibitem[Chickering(2002)]{chickering2002optimal}
David~Maxwell Chickering.
\newblock Optimal structure identification with greedy search.
\newblock \emph{Journal of machine learning research}, 3\penalty0
  (Nov):\penalty0 507--554, 2002.

\bibitem[Chickering(2020)]{chickering2020statistically}
Max Chickering.
\newblock Statistically efficient greedy equivalence search.
\newblock In \emph{Conference on Uncertainty in Artificial Intelligence}, pages
  241--249. PMLR, 2020.

\bibitem[Chickering et~al.(2004)Chickering, Heckerman, and
  Meek]{chickering2004large}
Max Chickering, David Heckerman, and Chris Meek.
\newblock Large-sample learning of bayesian networks is np-hard.
\newblock \emph{Journal of Machine Learning Research}, 5:\penalty0 1287--1330,
  2004.

\bibitem[Choi et~al.(2011)Choi, Tan, Anandkumar, and Willsky]{choi2011learning}
Myung~Jin Choi, Vincent~YF Tan, Animashree Anandkumar, and Alan~S Willsky.
\newblock Learning latent tree graphical models.
\newblock \emph{Journal of Machine Learning Research}, 12:\penalty0 1771--1812,
  2011.

\bibitem[Colombo et~al.(2012)Colombo, Maathuis, Kalisch, and
  Richardson]{colombo2012learning}
Diego Colombo, Marloes~H Maathuis, Markus Kalisch, and Thomas~S Richardson.
\newblock Learning high-dimensional directed acyclic graphs with latent and
  selection variables.
\newblock \emph{The Annals of Statistics}, pages 294--321, 2012.

\bibitem[Drton and Perlman(2007)]{drton2007multiple}
Mathias Drton and Michael~D Perlman.
\newblock Multiple testing and error control in gaussian graphical model
  selection.
\newblock \emph{Statistical Science}, 22\penalty0 (3):\penalty0 430--449, 2007.

\bibitem[Drton et~al.(2007)Drton, Sturmfels, and Sullivant]{drton2007algebraic}
Mathias Drton, Bernd Sturmfels, and Seth Sullivant.
\newblock Algebraic factor analysis: tetrads, pentads and beyond.
\newblock \emph{Probability Theory and Related Fields}, 138\penalty0
  (3-4):\penalty0 463--493, 2007.

\bibitem[Drton et~al.(2008)Drton, Massam, and Olkin]{drton2008moments}
Mathias Drton, H{\'e}l{\`e}ne Massam, and Ingram Olkin.
\newblock Moments of minors of wishart matrices.
\newblock \emph{The Annals of Statistics}, 36\penalty0 (5):\penalty0
  2261--2283, 2008.

\bibitem[Drton et~al.(2017)Drton, Lin, Weihs, and Zwiernik]{drton2017marginal}
Mathias Drton, Shaowei Lin, Luca Weihs, and Piotr Zwiernik.
\newblock Marginal likelihood and model selection for gaussian latent tree and
  forest models.
\newblock \emph{Bernoulli}, 23\penalty0 (2):\penalty0 1202--1232, 2017.

\bibitem[Eisen et~al.(1998)Eisen, Spellman, Brown, and
  Botstein]{eisen1998cluster}
Michael~B Eisen, Paul~T Spellman, Patrick~O Brown, and David Botstein.
\newblock Cluster analysis and display of genome-wide expression patterns.
\newblock \emph{Proceedings of the National Academy of Sciences}, 95\penalty0
  (25):\penalty0 14863--14868, 1998.

\bibitem[Friedman et~al.(2000)Friedman, Linial, Nachman, and
  Pe'er]{friedman2000using}
Nir Friedman, Michal Linial, Iftach Nachman, and Dana Pe'er.
\newblock Using bayesian networks to analyze expression data.
\newblock \emph{Journal of computational biology}, 7\penalty0 (3-4):\penalty0
  601--620, 2000.

\bibitem[Frot et~al.(2019)Frot, Nandy, and Maathuis]{frot2019robust}
Benjamin Frot, Preetam Nandy, and Marloes~H Maathuis.
\newblock Robust causal structure learning with some hidden variables.
\newblock \emph{Journal of the Royal Statistical Society}, 2019.

\bibitem[Gretton et~al.(2007)Gretton, Fukumizu, Teo, Song, Sch{\"o}lkopf,
  Smola, et~al.]{gretton2007kernel}
Arthur Gretton, Kenji Fukumizu, Choon~Hui Teo, Le~Song, Bernhard Sch{\"o}lkopf,
  Alexander~J Smola, et~al.
\newblock A kernel statistical test of independence.
\newblock In \emph{Nips}, volume~20, pages 585--592. Citeseer, 2007.

\bibitem[Hauser and B{\"u}hlmann(2012)]{hauser2012characterization}
Alain Hauser and Peter B{\"u}hlmann.
\newblock Characterization and greedy learning of interventional markov
  equivalence classes of directed acyclic graphs.
\newblock \emph{The Journal of Machine Learning Research}, 13\penalty0
  (1):\penalty0 2409--2464, 2012.

\bibitem[Heinze-Deml et~al.(2018)Heinze-Deml, Maathuis, and
  Meinshausen]{heinze2018causal}
Christina Heinze-Deml, Marloes~H Maathuis, and Nicolai Meinshausen.
\newblock Causal structure learning.
\newblock \emph{Annual Review of Statistics and Its Application}, 5:\penalty0
  371--391, 2018.

\bibitem[Higham et~al.(2007)Higham, Kalna, and Kibble]{higham2007spectral}
Desmond~J Higham, Gabriela Kalna, and Milla Kibble.
\newblock Spectral clustering and its use in bioinformatics.
\newblock \emph{Journal of computational and applied mathematics}, 204\penalty0
  (1):\penalty0 25--37, 2007.

\bibitem[Hosseini and Lee(2016)]{hosseini2016learning}
Mohammad~Javad Hosseini and Su-In Lee.
\newblock Learning sparse gaussian graphical models with overlapping blocks.
\newblock In \emph{Advances in Neural Information Processing Systems}, pages
  3808--3816, 2016.

\bibitem[Kummerfeld and Ramsey(2016)]{kummerfeld2016causal}
Erich Kummerfeld and Joseph Ramsey.
\newblock Causal clustering for 1-factor measurement models.
\newblock In \emph{Proceedings of the 22nd ACM SIGKDD international conference
  on knowledge discovery and data mining}, pages 1655--1664, 2016.

\bibitem[Kummerfeld et~al.(2014)Kummerfeld, Ramsey, Yang, Spirtes, and
  Scheines]{kummerfeld2014causal}
Erich Kummerfeld, Joe Ramsey, Renjie Yang, Peter Spirtes, and Richard Scheines.
\newblock Causal clustering for 2-factor measurement models.
\newblock In \emph{Joint European Conference on Machine Learning and Knowledge
  Discovery in Databases}, pages 34--49. Springer, 2014.

\bibitem[Leung and Drton(2018)]{leung2018algebraic}
Dennis Leung and Mathias Drton.
\newblock Algebraic tests of general gaussian latent tree models.
\newblock \emph{Advances in Neural Information Processing Systems}, 31, 2018.

\bibitem[Lieberman-Aiden et~al.(2009)Lieberman-Aiden, van Berkum, Williams,
  Imakaev, Ragoczy, Telling, Amit, Lajoie, Sabo, Dorschner, Sandstrom,
  Bernstein, Bender, Groudine, Gnirke, Stamatoyannopoulos, Mirny, Lander, and
  Dekker]{Lieberman}
Erez Lieberman-Aiden, Nynke~L. van Berkum, Louise Williams, Maxim Imakaev,
  Tobias Ragoczy, Agnes Telling, Ido Amit, Bryan~R. Lajoie, Peter~J. Sabo,
  Michael~O. Dorschner, Richard Sandstrom, Bradley Bernstein, M.A. Bender, Mark
  Groudine, Andreas Gnirke, John Stamatoyannopoulos, Leonid~A. Mirny, Eric~S.
  Lander, and Job Dekker.
\newblock Comprehensive mapping of long-range interactions reveals folding
  principles of the human genome.
\newblock \emph{Science}, 326:\penalty0 289--293, 2009.

\bibitem[Nandy et~al.(2018)Nandy, Hauser, and Maathuis]{nandy2018high}
Preetam Nandy, Alain Hauser, and Marloes~H Maathuis.
\newblock High-dimensional consistency in score-based and hybrid structure
  learning.
\newblock \emph{The Annals of Statistics}, 46\penalty0 (6A):\penalty0
  3151--3183, 2018.

\bibitem[Peters and B{\"u}hlmann(2014)]{peters2014identifiability}
Jonas Peters and Peter B{\"u}hlmann.
\newblock Identifiability of gaussian structural equation models with equal
  error variances.
\newblock \emph{Biometrika}, 101\penalty0 (1):\penalty0 219--228, 2014.

\bibitem[Pircalabelu and Claeskens(2020)]{pircalabelu2020community}
Eugen Pircalabelu and Gerda Claeskens.
\newblock Community-based group graphical lasso.
\newblock \emph{J. Mach. Learn. Res.}, 21:\penalty0 64--1, 2020.

\bibitem[Reisach et~al.(2021)Reisach, Seiler, and Weichwald]{reisach2021beware}
Alexander~Gilbert Reisach, Christof Seiler, and Sebastian Weichwald.
\newblock {Beware of the Simulated DAG! Causal Discovery Benchmarks May Be Easy
  to Game}.
\newblock In A.~Beygelzimer, Y.~Dauphin, P.~Liang, and J.~Wortman Vaughan,
  editors, \emph{Advances in Neural Information Processing Systems}, 2021.

\bibitem[Richardson and Spirtes(2002)]{richardson2002ancestral}
Thomas Richardson and Peter Spirtes.
\newblock Ancestral graph markov models.
\newblock \emph{The Annals of Statistics}, 30\penalty0 (4):\penalty0 962--1030,
  2002.

\bibitem[Robins et~al.(2000)Robins, Hernan, and Brumback]{robins2000marginal}
James~M Robins, Miguel~Angel Hernan, and Babette Brumback.
\newblock Marginal structural models and causal inference in epidemiology,
  2000.

\bibitem[Sachs et~al.(2005)Sachs, Perez, Pe'er, Lauffenburger, and
  Nolan]{sachs2005causal}
Karen Sachs, Omar Perez, Dana Pe'er, Douglas~A Lauffenburger, and Garry~P
  Nolan.
\newblock Causal protein-signaling networks derived from multiparameter
  single-cell data.
\newblock \emph{Science}, 308\penalty0 (5721):\penalty0 523--529, 2005.

\bibitem[Shiers et~al.(2016)Shiers, Zwiernik, Aston, and
  Smith]{shiers2016correlation}
Nathaniel Shiers, Piotr Zwiernik, John~AD Aston, and James~Q Smith.
\newblock The correlation space of gaussian latent tree models and model
  selection without fitting.
\newblock \emph{Biometrika}, 103\penalty0 (3):\penalty0 531--545, 2016.

\bibitem[Shimizu et~al.(2006)Shimizu, Hoyer, Hyv{\"a}rinen, Kerminen, and
  Jordan]{shimizu2006linear}
Shohei Shimizu, Patrik~O Hoyer, Aapo Hyv{\"a}rinen, Antti Kerminen, and Michael
  Jordan.
\newblock A linear non-gaussian acyclic model for causal discovery.
\newblock \emph{Journal of Machine Learning Research}, 7\penalty0 (10), 2006.

\bibitem[Shimizu et~al.(2009)Shimizu, Hoyer, and
  Hyv{\"a}rinen]{shimizu2009estimation}
Shohei Shimizu, Patrik~O Hoyer, and Aapo Hyv{\"a}rinen.
\newblock Estimation of linear non-gaussian acyclic models for latent factors.
\newblock \emph{Neurocomputing}, 72\penalty0 (7-9):\penalty0 2024--2027, 2009.

\bibitem[Silva et~al.(2006)Silva, Scheines, Glymour, Spirtes, and
  Chickering]{silva2006learning}
Ricardo Silva, Richard Scheines, Clark Glymour, Peter Spirtes, and
  David~Maxwell Chickering.
\newblock Learning the structure of linear latent variable models.
\newblock \emph{Journal of Machine Learning Research}, 7\penalty0 (2), 2006.

\bibitem[Solus et~al.(2021)Solus, Wang, and Uhler]{solus2021consistency}
Liam Solus, Yuhao Wang, and Caroline Uhler.
\newblock Consistency guarantees for greedy permutation-based causal inference
  algorithms.
\newblock \emph{Biometrika}, 2021.

\bibitem[Spirtes(2001)]{spirtes2001anytime}
Peter Spirtes.
\newblock An anytime algorithm for causal inference.
\newblock In \emph{International Workshop on Artificial Intelligence and
  Statistics}, pages 278--285. PMLR, 2001.

\bibitem[Spirtes et~al.(2000)Spirtes, Glymour, Scheines, and
  Heckerman]{spirtes2000causation}
Peter Spirtes, Clark~N Glymour, Richard Scheines, and David Heckerman.
\newblock \emph{Causation, prediction, and search}.
\newblock MIT press, 2000.

\bibitem[Spirtes(2013)]{spirtes2013calculation}
Peter~L Spirtes.
\newblock Calculation of entailed rank constraints in partially non-linear and
  cyclic models.
\newblock \emph{arXiv preprint arXiv:1309.7004}, 2013.

\bibitem[Stadhouders et~al.(2019)Stadhouders, Filion, and
  Graf]{stadhouders2019transcription}
Ralph Stadhouders, Guillaume~J Filion, and Thomas Graf.
\newblock Transcription factors and 3d genome conformation in cell-fate
  decisions.
\newblock \emph{Nature}, 569\penalty0 (7756):\penalty0 345--354, 2019.

\bibitem[Sullivant et~al.(2010)Sullivant, Talaska, and
  Draisma]{sullivant2010trek}
Seth Sullivant, Kelli Talaska, and Jan Draisma.
\newblock Trek separation for gaussian graphical models.
\newblock \emph{The Annals of Statistics}, 38\penalty0 (3):\penalty0
  1665--1685, 2010.

\bibitem[Tan et~al.(2015)Tan, Witten, and Shojaie]{tan2015cluster}
Kean~Ming Tan, Daniela Witten, and Ali Shojaie.
\newblock The cluster graphical lasso for improved estimation of gaussian
  graphical models.
\newblock \emph{Computational statistics \& data analysis}, 85:\penalty0
  23--36, 2015.

\bibitem[Uhler and Shivashankar(2017)]{Shiva_review}
Caroline Uhler and G.V. Shivashankar.
\newblock Regulation of genome organization and gene expression by nuclear
  mechanotransduction.
\newblock \emph{Nature Reviews Molecular Cell Biology}, 18:\penalty0 717--727,
  2017.

\bibitem[Uhler et~al.(2013)Uhler, Raskutti, B{\"u}hlmann, and
  Yu]{uhler2013geometry}
Caroline Uhler, Garvesh Raskutti, Peter B{\"u}hlmann, and Bin Yu.
\newblock Geometry of the faithfulness assumption in causal inference.
\newblock \emph{The Annals of Statistics}, pages 436--463, 2013.

\bibitem[van~der Linden and Hambleton(2013)]{van2013handbook}
Wim~J van~der Linden and Ronald~K Hambleton.
\newblock \emph{Handbook of modern item response theory}.
\newblock Springer Science \& Business Media, 2013.

\bibitem[Wang et~al.(2017)Wang, Solus, Yang, and Uhler]{wang2017permutation}
Yuhao Wang, Liam Solus, Karren~D. Yang, and Caroline Uhler.
\newblock Permutation-based causal inference algorithms with interventions.
\newblock In \emph{Neural Information Processing Systems}, volume~31, 2017.

\bibitem[Wishart(1928)]{wishart1928sampling}
John Wishart.
\newblock Sampling errors in the theory of two factors.
\newblock \emph{British Journal of Psychology}, 19\penalty0 (2):\penalty0 180,
  1928.

\bibitem[Xie et~al.(2020)Xie, Cai, Huang, Glymour, Hao, and
  Zhang]{xie2020generalized}
Feng Xie, Ruichu Cai, Biwei Huang, Clark Glymour, Zhifeng Hao, and Kun Zhang.
\newblock Generalized independent noise condition for estimating latent
  variable causal graphs.
\newblock \emph{Advances in Neural Information Processing Systems},
  33:\penalty0 14891--14902, 2020.

\bibitem[Zhang et~al.(2011)Zhang, Peters, Janzing, and
  Sch{\"o}lkopf]{zhang2011kernel}
Kun Zhang, Jonas Peters, Dominik Janzing, and Bernhard Sch{\"o}lkopf.
\newblock Kernel-based conditional independence test and application in causal
  discovery.
\newblock In \emph{Proceedings of the Twenty-Seventh Conference on Uncertainty
  in Artificial Intelligence}, pages 804--813, 2011.

\end{thebibliography}

\newpage
\appendix


\section{Faithfulness assumptions are generic}\label{appendix:faithfulness-generic}

We first recall the assumptions from \rref{sec:identifiability}.

\clusterFaithfulness*
\parentTetradFaithfulness*
\latentAdjacencyFaithfulness*

\begin{prop}
\rref{assumption:cluster-tetrad-faithfulness} holds generically.
\end{prop}
\begin{proof}
Let $L_i = \pa(X_i)$ and $L_j = \pa(X_j)$.
By the triple child assumption, there exists some $X_u$ in the same cluster as $X_i$, and some $X_v$ in the same cluster as $X_j$.
Then any set which t-separates $\{ i, j \}$ and $\{ u, v \}$ must contain $L_i$ and $L_j$, since $i$ must be separated from $u$ and $j$ must be separated from $v$, respectively.
Therefore, by \rref{thm:trek-separation}, $\rank(\Sigma_{[ij],[uv]}) = 2$ generically, i.e., $t_{ij,uv} \neq 0$ generically.
\end{proof}

\begin{prop}
\rref{assumption:parent-tetrad-faithfulness} holds generically.
\end{prop}
\begin{proof}
Let $L_i = \pa(X_i) = \pa(X_j)$.
$L_i$ must have some other child $X_u$ by the triple child assumption.
Let $L_v$ be some child of $X_i$, and $X_v$ be some child of $L_v$.
Then any set which t-separates $\{ i, j \}$ and $\{ u, v \}$ must contain $L_i$ and $L_v$, since $i$ must be separated from $j$ and $i$ must be separated from $v$, respectively. Therefore, by \rref{thm:trek-separation}, $\rank(\Sigma_{[ij],[uv]}) = 2$ generically, i.e., $t_{ij,uv} \neq 0$ generically, which completes the proof.
\end{proof}

\begin{prop}
\rref{assumption:latent-adjacency-faithfulness} holds generically.
\end{prop}
\begin{proof}
If $X_i \to L_k$, then $X_i$ and $X_k$ are d-connected given $S_i, S$ for any $X_k \in \ch(L_k)$.
\cite{spirtes2000causation} establish that if two nodes are d-connected, then their partial correlation in a linear SEM is generically nonzero, proving the desired result.
\end{proof}

\section{Proof of \rref{thm:consistency}}\label{appendix:algorithm-consistency}

\consistency*
\begin{proof}
By \rref{assumption:cluster-tetrad-faithfulness} and \rref{assumption:parent-tetrad-faithfulness}, as long as at least two nodes are present from each cluster, if $t_{ij,uv} = 0$ for all pairs $u, v \in \{ i, j \}$ iff. $X_i$ and $X_j$ have the same latent parent, and neither node has any children.
Thus, if $i - j$  in $\cG$ in \rref{alg:find-bottom}, then $i$ and $j$ are in the same cluster.
Next, if $i$ and $j$ are both left in $R$ after the while loop, then they must be in the same cluster.
For sake of contradiction, suppose not, and let $L_i = \pa(X_i)$, $L_j = \pa(X_j)$.
Since $X_i$ remains, then by the double-parent assumption, there must also remain some other node $X_{i'}$ that is a child of $L_i$.
Similarly, there must remain some other node $X_{j'}$ that is a child of $L_j$.
However, then $|R| = 4$, a contradiction.
Therefore, the clustering output by \rref{alg:find-bottom} is a refinement of the true clustering.
Furthermore, if a node $i$ is upstream of the cluster $C_1$, then $i$ necessarily has a child, and thus $t_{ij,uv} \neq 0$ for some $j, u, v$.
Therefore, $i$ cannot be placed in any clique before the cluster $C_1$ is completely removed, and thus the ordering of cluster returned by \rref{alg:find-bottom} is topologically consistent.
By the double parent assumption, each cluster in $\pi$ from \rref{alg:find-bottom} has size at least 2. 
\rref{assumption:cluster-tetrad-faithfulness} ensures that two clusters are merged by \rref{alg:mergeCluster} iff. they have the same latent parent.
Finally, \rref{assumption:parent-tetrad-faithfulness} ensures that $X_i \to L_k$ in $\hat{\cG}$ if and only if $X_i \to L_k$ in $\cG$.
\end{proof}

\end{document}